\documentclass[journal]{IEEEtran}

\usepackage{ifpdf}
\usepackage{cite}
\usepackage{booktabs}
\usepackage{bm}
\ifCLASSINFOpdf
\else
  \usepackage[dvipdfmx]{graphicx}
\fi

\usepackage[cmex10,fleqn]{amsmath}
\usepackage{url}

\usepackage{amsthm}
\theoremstyle{definition}

\ifCLASSOPTIONcompsoc
	\usepackage[font=normalsize,labelfont=sf,textfont=sf]{subfig}
\else
	\usepackage[font=footnotesize]{subfig}
\fi

\usepackage{amsmath,amssymb}
\usepackage{bm}
\usepackage{color}
\usepackage{mathtools}

\DeclarePairedDelimiter\parens{\lparen}{\rparen}

\DeclarePairedDelimiter\brackets{[}{]}

\newtheorem{lemma}{Lemma}
\newtheorem{proposition}{Proposition}

\RequirePackage[normalem]{ulem} 

\begin{document}
\title{\LARGE Stochastic Geometry Analysis of Normalized SNR-Based Scheduling in Downlink Cellular Networks}

\author{Takuya~Ohto,~\IEEEmembership{}
Koji~Yamamoto,~\IEEEmembership{Member,~IEEE,}
Seong-Lyun~Kim,~\IEEEmembership{Member,~IEEE,}
Takayuki~Nishio,~\IEEEmembership{Member,~IEEE,}
and~Masahiro Morikura,~\IEEEmembership{Member,~IEEE}
\thanks{Manuscript received **. The associate editor coordinating the review of this letter and approving it for publication was **}
\thanks{T.\ Ohto, K.\ Yamamoto (corresponding author), T.\ Nishio, and M.\ Morikura are with Graduate School of Informatics, Kyoto University, Kyoto, 606-8501 Japan (e-mail: kyamamot@i.kyoto-u.ac.jp).}
\thanks{S.\ Kim is with the Radio Resource Management and Optimization Laboratory, School of Electrical and Electronic Engineering, Yonsei University, 50 Yonsei-ro, Seodaemun-gu, Seoul 120-749, Korea.}
\thanks{The present study was supported in part by a Grant-in-Aid for Scientific Research (C) (no.\@15K06062) from the Japan Society for the Promotion of Science (JSPS). }
\thanks{Digital Object Identifier ***}

\markboth{IEEE Communications Letters,~Vol.~, No.~, March 2011}%
{OHTO \MakeLowercase{\textit{et al.}}:Stochastic Geometry Analysis of Normalized SNR-Based Scheduling in Downlink Cellular Networks}

}
\maketitle

\begin{abstract}

The coverage probability and average data rate of normalized SNR-based scheduling in a downlink cellular network are derived
by modeling the locations of the base stations and users as two independent Poison point processes.
The scheduler selects the user with the largest instantaneous SNR normalized by the short-term average SNR.
In normalized SNR scheduling, the coverage probability when the desired signal experiences Rayleigh fading is shown to be given by a series of Laplace transforms of the probability density function of interference.
Also, a closed-form expression for the coverage probability is approximately achieved.
The results confirm that normalized SNR scheduling increases the coverage probability due to the multi-user diversity gain.
\end{abstract}
\IEEEpeerreviewmaketitle

\section{Introduction}
\IEEEPARstart{S}{tochastic} geometry enables tractable modeling and accurate analysis of cellular networks \cite{andrews2011tractable,yu2013downlink,elsawy2013stochastic,ElSawy2017CSTO}.
In particular, the coverage probability of a randomly chosen user depending on co-channel interference can be expressed in a closed form in a special case \cite{andrews2011tractable} and has been analyzed in a wide variety of scenarios \cite{elsawy2013stochastic,ElSawy2017CSTO}.
However, in terms of stochastic geometry analyses of cellular networks, where the locations of the base stations form a Poisson point process (PPP) as in \cite{andrews2011tractable,yu2013downlink}, as far as the authors know, no studies have been conducted into channel-adaptive scheduling.
In a single-cell environment, normalized SNR based scheduler and proportional fair (PF) scheduler \cite{viswanath2002opportunistic} were analyzed in \cite{yang2006performance} and \cite{Choi2007VT}, respectively.
In a multi-cell environment, normalized SNR based scheduler was analyzed assuming that interference is independent of time and the locations of users in \cite{Blaszczyszyn2009ComNet}.
In hexagonal cell arrangements with users forming a PPP, channel-adaptive scheduling based on the normalized signal-to-interference-plus-noise power ratio (SINR) was analyzed in \cite{Garcia-Morales2015CIT}.
Note that these studies \cite{Blaszczyszyn2009ComNet,Garcia-Morales2015CIT} evaluated average throughput of the network, not coverage probability.

In this letter, using the framework of stochastic geometry analysis of cellular networks \cite{andrews2011tractable,yu2013downlink}, we derive coverage probability and average data rate in downlink cellular networks with channel-adaptive scheduling, where users experience Rayleigh fading.
As in \cite{yang2006performance,Choi2007VT,Blaszczyszyn2009ComNet}, for ease of analysis, 
we employ a scheduling scheme based on the instantaneous SNR normalized by the {\em short-term} average SNR, while the coverage probability is defined as the probability that users achieve a target SINR.
%
%
Henceforth, we refer to this scheme as {\em normalized SNR scheduling}.
Normalized SNR scheduling achieves the same temporal fairness among users as round-robin (RR) scheduling \cite{yang2006performance},
and is
similar to PF scheduler, except that this scheme is based on the normalized SNR rather than the normalized data rate.
Note that normalized SNR scheduling is equivalent to PF scheduling if the data rate is proportional to the SNR \cite{Choi2007VT}.
We clarify that when the desired signal experiences Rayleigh fading in normalized SNR scheduling,
the coverage probability is given by a series of Laplace transforms of the probability density function (pdf) of interference.
We also evaluate the {\em scheduling gain} \cite{berggren2004asymptotically}, which is defined as the ratio of the average data rate of normalized SNR scheduling to that of RR scheduling.

 {\em Notation}:
$\mathbb{E}[\cdot]$ denotes the expectation operator,
$f_X(\cdot)$ denotes the pdf of a random variable $X$,
$F_X(\cdot)$ denotes the cumulative distribution function (cdf) of $X$,
$\mathcal{L}_X(\cdot)$ denotes the Laplace transform of the pdf of $X$,
and $\Gamma(\cdot)$ denotes the gamma function.

\section{System Model}\label{SM}
Our system model of a downlink cellular network consists of both base stations (BSs) and users, as was discussed in \cite{yu2013downlink}.
Note that the main difference between our system model and the model in \cite{andrews2011tractable} is in considering the user distribution.

The locations of the BSs are assumed to be distributed according to a homogeneous PPP $\Phi_\mathrm{b}$ with intensity $\lambda_\mathrm{b}$ on the Euclidean plane $\mathbb{R}^2$.
We assume that each BS operates with an identical transmission power $P_\mathrm{b}$.
The locations of the users  are also 
distributed according to a homogeneous PPP $\Phi_\mathrm{u}$ with intensity $\lambda_\mathrm{u}$.
Each user is 
associated with the nearest BS
\cite{yu2013downlink}.
That is, the cell area of each BS comprises a Voronoi tessellation, 
as shown in Fig.~\ref{fig:coverage101}.
We consider 
one resource block to assign, and
each BS is assumed to serve only one user in a resource block at any given time.

\begin{figure}[t]
\centering
   \includegraphics[width=.5\linewidth]{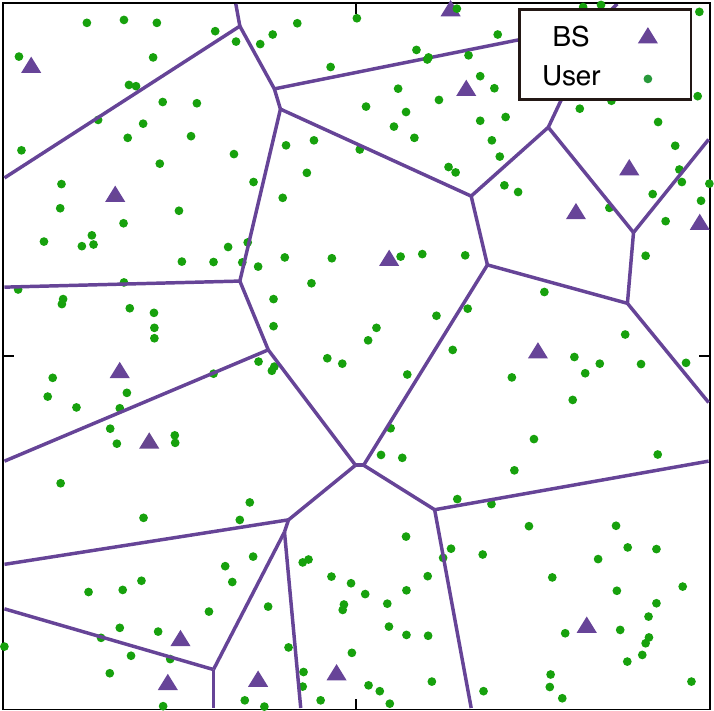}
  \caption{Deployment of BSs and users.}
  \label{fig:coverage101}
\end{figure}

The scheduler selects the user with the largest instantaneous SNR normalized by the short-term average SNR, which is the average value of the instantaneous SNR over a period when variation in the distance between the user and its associated BS is negligible as in \cite{yang2006performance,viswanath2002opportunistic}.
We assume that users experience quasi-static Rayleigh fading, i.e., the channel gain is constant over a time slot.
As in \cite{yang2006performance,Blaszczyszyn2009ComNet,viswanath2002opportunistic}, we assume perfect channel estimation at the beginning of each time slot.
The instantaneous SNR of the user at a distance $r$ from the associated BS is, therefore, $P_\mathrm{b}hr^{-\alpha}/\sigma_\mathrm{n}^2$, where $h$ represents the fading gain, which is an exponentially distributed random variable with unit mean, i.e., $h\sim\exp(1)$, $\alpha>2$ represents the path loss exponent,
and $\sigma_\mathrm{n}^2$ represents the noise power.

We consider two scenarios: Scenario 1 where all of the BSs continually transmit signals independently of the number of associated users; and
Scenario 2 where a BS that has no user to serve does not transmit any signals.
The coverage probability and scheduling gain in the first scenario are discussed in Section \ref{scenario1}, and those in the second scenario are given in Section \ref{scenario2}.

\section{Coverage Probability and Average Data Rate}\label{CP}
The coverage probability of the tagged user can be derived as a function of the BS and user densities when normalized SNR scheduling is applied.
The coverage probability of a tagged user is the probability that the tagged user can achieve a target SINR $\theta$ when the tagged user is scheduled, which is defined as
$p_\mathrm{c}(\theta)\coloneqq \mathbb{P}(\mathit{SINR}>\theta)$.
That is, the coverage probability is defined as the complementary cumulative distribution function (ccdf) of the instantaneous SINR.
Let the number of users in the cell of the BS that serves the tagged user, except for the tagged user, be a random variable $N$, and let the tagged user be located at a random distance $R$ from the associated BS.
Conditioning on $N=n$ and $R=r$, the instantaneous SINR of the tagged user can be written as:
\begin{align}
\mathit{SINR}=\frac{P_\mathrm{b}h_{n+1:n+1}r^{-\alpha}}{\sigma^{2}_\mathrm{n}+I},
\end{align}
where $I$ denotes the aggregate interference received at the tagged user, and $h_{n+1:n+1}$ denotes the fading gain of the tagged user when the tagged user is scheduled, as explained later in this letter.
The tagged user is assumed to be located at $o$, and the location of the associated BS is denoted by $b_o$.
In Scenario 1, because all of the BSs continually transmit signals,
the aggregate interference $I$ is given by:
\begin{align}
I=P_\mathrm{b}\sum_{b\in\Phi_\mathrm{b}\backslash \{b_o\}}g_bd_b^{-\alpha},
\end{align}
where $g_b$ is the fading gain between the tagged user and the BS at $b$, and $d_b$ is the distance from the tagged user to the BS at $b$.
In Scenario 2,  there might be a BS with no associated user, and thus,
the aggregate interference is received only from the active BSs, i.e., BSs that have at least one user to serve.

When the tagged user is scheduled, the coverage probability of the tagged user can be approximately given as: 
\begin{align}
\MoveEqLeft p_\mathrm{c}(\theta)
=\mathbb{E}_{R,N}[\mathbb{P}(\mathit{SINR}>\theta\mid r, n)]\nonumber\\
&\simeq \sum_{n=0}^{\infty}\int_{0}^{\infty}\mathbb{P}(\mathit{SINR}>\theta\mid r,n)f_{R}(r)f_{N}(n)\,\mathrm{d}r,
\label{eq:coverage_origin}
\end{align}
where we assume that $N$ and $R$ are independent, and $\mathbb{P}(\mathit{SINR}>\theta\mid r, n)$ denotes the coverage probability conditioning on $R=r$ and $N=n$, which is given by the ccdf of  $h_{n+1:n+1}$, $\mathbb{P}(h_{n+1:n+1}>r^\alpha\theta(\sigma_\mathrm{n}^2+I)/P_\mathrm{b})$.
The accuracy of the approximation is discussed in Section \ref{NE}.
According to \cite{andrews2011tractable,haenggi2012stochastic}, the pdf of $R$ is given by:
\begin{align}
f_R(r)=2\pi\lambda_\mathrm{b} r\mathrm{e}^{-\pi \lambda_\mathrm{b} r^2}.
\label{eq:pdf_r}
\end{align}
In addition, $f_N(n)$ denotes the probability mass function (pmf) of $N$.
From Slivnyak's theorem \cite{haenggi2012stochastic}, the locations of the other users follow the reduced Palm distribution with $\Phi_\mathrm{u}$.
According to \cite[Lemma 1]{akoum2013interference}, the pmf of $N$ is given by:
\begin{align}
f_N(n) 
&=\frac{\Gamma(n+c+1)(\lambda_\mathrm{u}/c\lambda_\mathrm{b})^n}{\Gamma(n+1)\Gamma(c+1)(\lambda_\mathrm{u}/c\lambda_\mathrm{b}+1)^{n+c+1}},
\label{eq:pmf_n}
\end{align}
where $c=7/2$.

\addtocounter{equation}{3}
\begin{figure*}[!t]
\vspace{-7mm}
\begin{align}
p_\mathrm{c}(\theta)&\simeq
\pi\lambda_\mathrm{b}\sum_{n=0}^{\infty}
\frac{\Gamma(n+c+1)\parens*{\frac{\lambda_\mathrm{u}}{c\lambda_\mathrm{b}}}^n}{\Gamma(n+1)\Gamma(c+1)\parens*{\frac{\lambda_\mathrm{u}}{c\lambda_\mathrm{b}}+1}^{n+c+1}}
\sum_{k=1}^{n+1}\tbinom{n+1}{k} (-1)^{k+1}\int_{0}^{\infty}\mathrm{e}^{-\pi\lambda_\mathrm{b}v(1+\rho(k\theta,\alpha))- \frac{kv^{\alpha/2}\theta\sigma_\mathrm{n}^2}{P_\mathrm{b}}}
\,\mathrm{d}v
\label{eq:theorem101}\\
\rho(k\theta,\alpha) &= (k\theta)^{2/\alpha} \int_{(k\theta)^{-2/\alpha}}^{\infty}\frac{1}{1+u^{\alpha/2}}\,\mathrm{d}u
\label{eq:theorem102}
\end{align}
\vspace{-5mm}
\hrulefill
\end{figure*}
\addtocounter{equation}{-5}

To evaluate \eqref{eq:coverage_origin},
we
derive the cdf of the fading gain of the tagged user when the tagged user is scheduled, $h_{n+1:n+1}$.
Let $i=1,\ldots, n+1$ be the index of users in the cell of BS $b_o$.
Letting the fading gain of $i$th user be denoted by $h_i$, 
$\{h_i\}$ are independent and identically distributed (i.i.d.) exponential random variables with unit mean.
Letting the distance from $i$th user to the associated BS be denoted by $r_i$,
the instantaneous SNR of the $i$th user, $\gamma_i$, is given as $\gamma_i=P_\mathrm{b}h_ir_i^{-\alpha}/\sigma_\mathrm{n}^2$, and its short-term average SNR, $\bar{\gamma}_i$, is given as $\bar{\gamma}_i\coloneqq\mathbb{E}[\gamma_i]=P_\mathrm{b}r_i^{-\alpha}/\sigma_\mathrm{n}^2$.
In normalized SNR scheduling, the scheduler selects the user with the largest instantaneous SNR normalized by the short-term average SNR.
We obtain $\gamma_i/\bar{\gamma}_i=h_i$.
That is, in normalized SNR scheduling, the scheduler selects the user with the largest fading gain.
The order statistics \cite{david1981order} obtained by arranging the values of $n+1$ $h_{i}$'s in increasing order of magnitude
are denoted as:
$h_{1:n+1}<h_{2:n+1}<\cdots<h_{n+1:n+1}$, and
the fading gain of the tagged user when the tagged user is scheduled is denoted by $h_{n+1:n+1}$.
The cdf of $h_{n+1:n+1}$ can be formulated by the same method as the calculation of the selection combiner output
because
in selection combining, 
the combiner output SNR is the largest SNR of all the branches.
According to \cite[\S7.2.2]{goldsmith2005wireless}, the cdf of $h_{n+1:n+1}$ can be written as:
\begin{align}
 F_{h_{n+1:n+1}}(x)=(1-\mathrm{e}^{-x})^{n+1}.
\label{eq:cdf_os}
\end{align}

\begin{lemma}
 \label{lemma101}
 When the desired signal experiences Rayleigh fading in normalized SNR scheduling, the coverage probability conditioning on $R=r$ and $N=n$, $\mathbb{P}(\mathit{SINR}>\theta \mid r,n)$, is given by a series of Laplace transforms of the pdf of interference $I$ as in \eqref{eq:appen_laplus}.
\end{lemma}

\begin{proof}
We have:
\begin{align}
\MoveEqLeft \mathbb{P}(\mathit{SINR}>\theta\mid r,n)
=\mathbb{P}(h_{n+1:n+1}>r^\alpha\theta(\sigma_\mathrm{n}^2+I)/P_\mathrm{b})\nonumber\\
={}& \mathbb{E}_I\brackets*{
1-
\sum_{k=0}^{n+1}\tbinom{n+1}{k}(-1)^{k}\mathrm{e}^{- \frac{kr^{\alpha}\theta\sigma_\mathrm{n}^2}{P_\mathrm{b}}}
\mathrm{e}^{- \frac{kr^{\alpha}\theta I}{P_\mathrm{b}}}
}
\label{eq:appen_2}\\
={}& 1-
\sum_{k=0}^{n+1}\tbinom{n+1}{k}(-1)^{k}\mathrm{e}^{- \frac{kr^{\alpha}\theta\sigma_\mathrm{n}^2}{P_\mathrm{b}}}
\mathbb{E}_I
\brackets*{
\mathrm{e}^{- \frac{kr^{\alpha}\theta I}{P_\mathrm{b}}}
}
\nonumber\\
={}& 
\sum_{k=1}^{n+1}\tbinom{n+1}{k}(-1)^{k+1}\mathrm{e}^{- \frac{kr^{\alpha}\theta\sigma_\mathrm{n}^2}{P_\mathrm{b}}}\mathcal{L}_I(kr^{\alpha}\theta/P_\mathrm{b}). 
\label{eq:appen_laplus}
\end{align}
The important point is that the coverage probability is given by an expectation of a series of exponential functions of interference $I$ as in \eqref{eq:appen_2}
because the cdf of the fading gain of the scheduled user \eqref{eq:cdf_os} can be written as a series of exponential functions from the binomial theorem, i.e.,
$ F_{h_{n+1:n+1}}(x)
=\sum_{k=0}^{n+1}\tbinom{n+1}{k}(-1)^{k}\mathrm{e}^{-kx}.$
Note that the reason for assuming Rayleigh fading for the desired signal is its tractability as discussed in \cite[Lemma 1]{andrews2011tractable}.
\end{proof}

\subsection{Scenario 1: Interference from All BSs}\label{scenario1}
In this scenario, all BSs continually transmit signals regardless of the number of associated users, and the tagged user receives interference from all BSs, except for its associated BS.
Substituting \eqref{eq:pdf_r}, \eqref{eq:pmf_n}, and \eqref{eq:appen_laplus} into \eqref{eq:coverage_origin}, 
we obtain $p_\mathrm{c}(\theta)$, as follows:
\begin{proposition}
\label{theorem101}
When interference experiences Rayleigh fading in normalized SNR scheduling, the coverage probability
is given by \eqref{eq:theorem101} and \eqref{eq:theorem102}.

\end{proposition}

\addtocounter{equation}{2}


\begin{proof}

According to \cite[Theorem 2]{andrews2011tractable},
$\mathcal{L}_I(kr^{\alpha}\theta/P_\mathrm{b})$ can be written as:
\begin{align}
\mathcal{L}_I(kr^{\alpha}\theta/P_\mathrm{b})&=
\exp(-\pi r^2 \lambda_\mathrm{b} \rho(k\theta,\alpha)), \\
\rho(k\theta,\alpha) &= (k\theta)^{2/\alpha} \int_{(k\theta)^{-2/\alpha}}^{\infty}\frac{1}{1+u^{\alpha/2}}\,\mathrm{d}u.
\end{align}
By using the substitution $r^2 \rightarrow v$, 
\begin{align}
p_\mathrm{c}(\theta)&\simeq
\pi\lambda_\mathrm{b}\sum_{n=0}^{\infty}
f_{N}(n)
\sum_{k=1}^{n+1}\tbinom{n+1}{k}(-1)^{k+1}\nonumber\\
&\hspace{-2mm}\cdot \int_{0}^{\infty}\mathrm{e}^{-\pi\lambda_\mathrm{b}v(1+\rho(k\theta,\alpha))- \frac{kv^{\alpha/2}\theta\sigma_\mathrm{n}^2}{P_\mathrm{b}}}
\,\mathrm{d}v.  \qedhere
\end{align}

\end{proof}

%
%


To simplify this expression, we consider a special case
as in \cite{andrews2011tractable}.
Assuming $\sigma_\mathrm{n}^2=0$ and $\alpha=4$, 
the expression of the coverage probability can be simplified to:
\begin{align}
p_\mathrm{c}(\theta)\simeq
\sum_{n=0}^{\infty}
f_{N}(n)
\sum_{k=1}^{n+1}\frac{\tbinom{n+1}{k}(-1)^{k+1}}{1+\sqrt{k\theta}\arctan\sqrt{k\theta}}.
\label{eq:cover_p303}
\end{align}
In this case, the coverage probability depends only on the target SINR $\theta$ and the ratio of the density of users to that of BSs, $\lambda_\mathrm{u}/\lambda_\mathrm{b}$.

For ease of performance evaluation, when $\lambda_\mathrm{u}/\lambda_\mathrm{b}$ is an integer, by roughly treating that $f_N(\lambda_\mathrm{u}/\lambda_\mathrm{b})=1$ and $f_N(n)=0$ for $n \neq \lambda_\mathrm{u}/\lambda_\mathrm{b}$ in (\ref{eq:cover_p303})\footnote{When $\lambda_\mathrm{u}/\lambda_\mathrm{b}$ is an integer, both $\lambda_\mathrm{u}/\lambda_\mathrm{b}-1$ and $\lambda_\mathrm{u}/\lambda_\mathrm{b}$ are the modes of $N$. We focus on the latter mode $\lambda_\mathrm{u}/\lambda_\mathrm{b}$ and treat $f_N(\lambda_\mathrm{u}/\lambda_\mathrm{b})=1$.}, we propose to approximate (\ref{eq:cover_p303}) to a closed form as:
\begin{align}
p_\mathrm{c}(\theta)\simeq
\sum_{k=1}^{\lambda_\mathrm{u}/\lambda_\mathrm{b}+1}\frac{\tbinom{\lambda_\mathrm{u}/\lambda_\mathrm{b}+1}{k}(-1)^{k+1}}{1+\sqrt{k\theta}\arctan\sqrt{k\theta}}.
\label{eq:further_approx}
\end{align}
The accuracy of the approximation is discussed in Section IV.

Here, we evaluate the scheduling gain of normalized SNR scheduling.
Based on the idea described in \cite{berggren2004asymptotically}, the scheduling gain is defined as $G(\lambda_\mathrm{b},\lambda_\mathrm{u})\coloneqq\tau_\mathrm{s}(\lambda_\mathrm{b},\lambda_\mathrm{u})/\tau_\mathrm{r}(\lambda_\mathrm{b})$, where $\tau_\mathrm{s}(\lambda_\mathrm{b},\lambda_\mathrm{u})$ and $\tau_\mathrm{r}(\lambda_\mathrm{b})$ denote average data rate of normalized SNR scheduling and RR scheduling, respectively.
The average data rate is defined as the mean rate in units of nats/Hz for the tagged user, which is assumed to achieve the Shannon bound at its instantaneous SINR, and $\tau_\mathrm{r}(\lambda_\mathrm{b})$ is derived in \cite{andrews2011tractable}.

We obtain $\tau_\mathrm{s}(\lambda_\mathrm{b},\lambda_\mathrm{u})$ as follows:
\begin{proposition}
\label{theorem201}
The average data rate of normalized SNR scheduling can be written as:
\begin{multline}
 \tau_\mathrm{s}(\lambda_\mathrm{b},\lambda_\mathrm{u}) \simeq
\pi\lambda_\mathrm{b}\sum_{n=0}^{\infty}
f_N(n)
\int_{0}^{\infty}\!\!\!
\int_{0}^{\infty}\sum_{k=1}^{n+1}\tbinom{n+1}{k}\\
\cdot(-1)^{k+1}\mathrm{e}^{-\pi\lambda_\mathrm{b}v(1+\rho(k(\mathrm{e}^t-1),\alpha))- \frac{kv^{\alpha/2}(\mathrm{e}^t-1)\sigma_\mathrm{n}^2}{P_\mathrm{b}}}
\,\mathrm{d}t\,\mathrm{d}v.
\label{eq:AER1}
\end{multline}
\end{proposition}
\begin{proof}
We have:
\begin{multline}
 \tau_\mathrm{s}(\lambda_\mathrm{b},\lambda_\mathrm{u})\coloneqq
\mathbb{E}[\ln(1+\mathit{SINR})]
\simeq \sum_{n=0}^{\infty}f_N(n)\int_{0}^{\infty}f_R(r)\\
\cdot\int_{0}^{\infty}\mathbb{P}\brackets*{\ln\parens*{1+\frac{P_\mathrm{b}h_{n+1:n+1}r^{-\alpha}}{\sigma_\mathrm{n}^2+I}}>t}\,\mathrm{d}t\,\mathrm{d}r.\nonumber
\end{multline} 
The rest of the proof is similar to that for \cite[Theorem 3]{andrews2011tractable} and Proposition \ref{theorem101}.
\end{proof}

In the case of $\sigma_\mathrm{n}^2=0$ and $\alpha=4$, the expression \eqref{eq:AER1} can be simplified to:
{\small
\begin{align}
 \sum_{n=0}^{\infty}f_N(n)
\int_{0}^{\infty}\sum_{k=1}^{n+1}
\frac{\tbinom{n+1}{k} (-1)^{k+1}}{1+\sqrt{k(\mathrm{e}^t-1)}\arctan\sqrt{k(\mathrm{e}^t-1)}}\,\mathrm{d}t.\nonumber
\end{align}
}

\subsection{Scenario 2: Interference Only from Active BSs}\label{scenario2}
We now consider the scenario where a BS that has no user to serve does not transmit any signals, as in \cite{yu2013downlink}.
In this scenario, the aggregate interference originates only from the active BSs, except for the BS associated with the tagged user.
Because the received power at the tagged user from the associated BS follows the identical distribution in both scenarios, and the density of the active BS is reduced from $\lambda_\mathrm{b}$ to $\lambda_\mathrm{b}(1-(1+\lambda_\mathrm{u}/c\lambda_\mathrm{b})^{-c})$ in Scenario 2,
according to \cite[Lemma 3]{yu2013downlink}, we obtain the coverage probability in Scenario 2 by substituting $(1-(1+\lambda_\mathrm{u}/c\lambda_\mathrm{b})^{-c})\rho(k\theta,\alpha)$ for $\rho(k\theta,\alpha)$ in \eqref{eq:theorem101}.

We then derive the scheduling gain $G(\lambda_\mathrm{b},\lambda_\mathrm{u})$
in Scenario 2.
We obtain the average data rate of normalized SNR scheduling by substituting $(1-(1+\lambda_\mathrm{u}/c\lambda_\mathrm{b})^{-c})\rho(k(\mathrm{e}^t-1),\alpha)$ for $\rho(k(\mathrm{e}^t-1),\alpha)$ in \eqref{eq:AER1}.
We also obtain the average data rate of RR scheduling by using \cite[Lemma 3]{yu2013downlink}.

\section{Numerical Examples}\label{NE}
We investigate the coverage probability and scheduling gain through numerical examples.
To confirm the impact of channel-adaptive scheduling on the performance of cellular networks,
we compare the coverage probability of the scheduled user in normalized SNR scheduling with that of a randomly scheduled user \cite{andrews2011tractable} with RR scheduler.

\begin{figure}[t]
\centering
 \includegraphics[width=.8\linewidth]{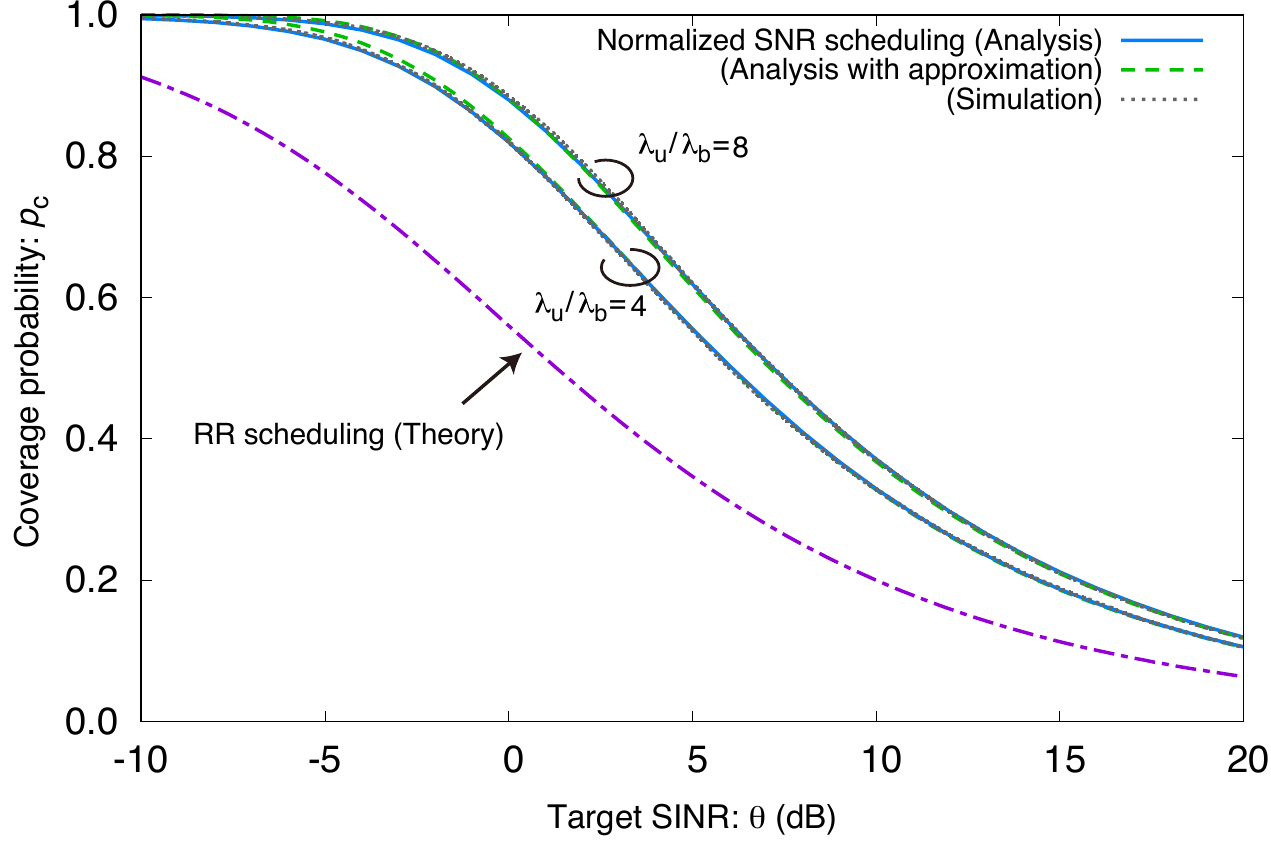}
  \caption{Coverage probability $p_\mathrm{c}$ for $\sigma_{\mathrm{n}}^2=0$, and $\alpha=4$. Analysis: \eqref{eq:cover_p303}, analysis with approximation: (\ref{eq:further_approx}).}
  \label{fig:coverage402}
\end{figure}

\begin{figure}[t]
\centering
   \includegraphics[width=.8\linewidth]{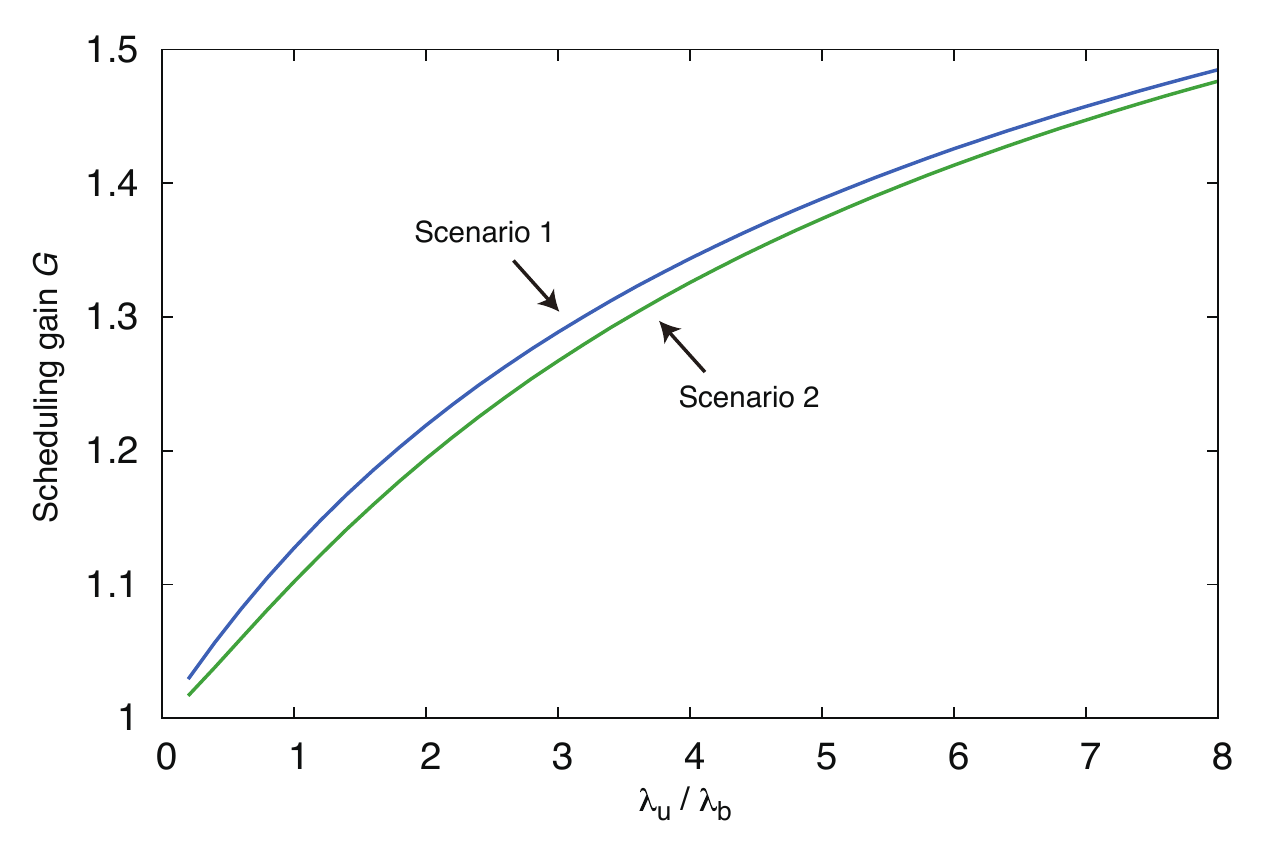}
  \caption{Scheduling gain $G$ for $\sigma_{\mathrm{n}}^2=0$, and $\alpha=4$.}
  \label{fig:s_gain}
\end{figure}

Fig.~\ref{fig:coverage402} shows the coverage probability in Scenario 1,
when $\sigma_{\mathrm{n}}^2=0$ and $\alpha=4$, \eqref{eq:cover_p303}, its approximation (\ref{eq:further_approx}), and its Monte Carlo simulation result.
In normalized SNR scheduling, the coverage probability increases along with $\lambda_\mathrm{u}/\lambda_\mathrm{b}$ according to 
multi-user diversity gain. 
In addition, simulation results coincide with \eqref{eq:cover_p303} and (\ref{eq:further_approx}).
 Thus, the accuracy of the approximation in (\ref{eq:coverage_origin}) and that in (\ref{eq:further_approx}) are validated.

Fig.~\ref{fig:s_gain} shows the scheduling gain in each scenario when $\sigma_{\mathrm{n}}^2=0$ and $\alpha=4$.
The scheduling gain increases along with $\lambda_\mathrm{u}/\lambda_\mathrm{b}$ due to the multi-user diversity gain.
In the limit $\lambda_\mathrm{u}/\lambda_\mathrm{b}\rightarrow \infty$, the scheduling gain in Scenario 2 would converge to that in Scenario 1 because the density of the active BSs, $\lambda_\mathrm{b}(1-(1+\lambda_\mathrm{u}/c\lambda_\mathrm{b})^{-c})$, approaches to $\lambda_\mathrm{b}$. 

\section{Conclusion}\label{CON}

This letter derived the coverage probability and average data rate of normalized SNR scheduler as a basic channel-adaptive scheduling scheme in cellular networks.
When the desired signal experiences Rayleigh fading, 
the coverage probability of normalized SNR scheduling is given by a series of Laplace transforms of the pdf of interference.
In the case of $\alpha=4$ and no noise, the coverage probability is only dependent on the target SINR and the ratio of the density of users to that of BSs,  $\lambda_\mathrm{u}/\lambda_\mathrm{b}$.
The closed-form expression for the outage probability is approximately given.
Numerical results confirmed that normalized SNR scheduling provides a higher coverage probability along with $\lambda_\mathrm{u}/\lambda_\mathrm{b}$ because of the multi-user diversity gain.
It was also confirmed that the scheduling gain increased along with $\lambda_\mathrm{u}/\lambda_\mathrm{b}$.

\end{document}